\documentclass[runningheads, a4paper,11pt]{llncs} 
\usepackage[margin=1.6in]{geometry}
\usepackage[english]{babel}

\usepackage[ruled,vlined,lined,boxed,commentsnumbered]{algorithm2e}
\usepackage[usenames,dvipsnames,svgnames,table]{xcolor}
\usepackage{wrapfig, blindtext}
\usepackage{amssymb, amsmath} 
\usepackage{mathtools}
\usepackage{tabularx,colortbl}
\usepackage{multirow}
\usepackage{url}

\newcommand{\qued}{\hfill$\square$}

\newcommand{\Sigmagap}{\Sigma_{gap}}
\newcommand{\bin}{\mathcal B(n,u)}
\newcommand{\Hg}{H_0(G)}

\newcommand{\bigO}{\mathcal{O}}

\begin{document}

\title{A Compressed-Gap Data-Aware Measure}

\author{Nicola Prezza\inst{1}}

\institute{Department of Mathematics and Computer Science, University of Udine, Italy \email{prezza.nicola@spes.uniud.it}}

\maketitle

\begin{abstract}

In this paper, we consider the problem of efficiently representing a set $S$ of $n$ items out of a universe $U=\{0,...,u-1\}$ while supporting a number of operations on it.
Let $G=g_1...g_n$ be the gap stream associated with $S$, $gap$ its bit-size when encoded with \emph{gap-encoding}, and $\Hg$ its empirical zero-order entropy. 
We prove that (1) $n\Hg\in o(gap)$ if $G$ is highly compressible, and (2) $n\Hg \leq n\log(u/n) + n \leq uH_0(S)$.
Let $d$ be the number of \emph{distinct} gap lengths between elements in $S$. We firstly propose a new space-efficient zero-order compressed representation of $S$ taking $n(\Hg+1)+\bigO(d\log u)$ bits of space. Then, we describe a fully-indexable dictionary that supports \emph{rank} and \emph{select} queries in $\bigO(\log(u/n)+\log\log u)$ time while requiring asymptotically the same space as the proposed compressed representation of $S$.
\keywords{dictionary problem, gap encoding, entropy, compression, rank, select}
\end{abstract}
	    
\section{Introduction and Related Work}\label{section:intro}

The \emph{dictionary problem} on \emph{set data} asks to maintain a (space-efficient) data structure called \emph{indexable dictionary} over a set $S = \{s_1, ..., s_n\} \subseteq \{0,...,u-1\} = U$, $s_1<s_2<...<s_n$, supporting efficiently a range of queries on $S$. In this problem, $U$ is an ordered set and is called \emph{universe}. As showed by Jacobson in his doctoral thesis~\cite{jacobson1988succinct}, a set of just two operations, \emph{rank} and \emph{select}, is sufficient and powerful enough in order to derive other fundamental functionalities desired from such a structure: \emph{member}, \emph{successor}, and \emph{predecessor}. $rank(S,x)$, with $x\in U$, is the number of elements in $S$ that are smaller than or equal to $x$. $select(S,i)$, where $0\leq i<n$, is the $i$-th smallest element in $S$. In this paper, we focus on \emph{fully-indexable dictionaries (FIDs)}, i.e. data structures supporting both rank and select operations efficiently.

Jacobson in~\cite{jacobson1988succinct} proposed a solution for this problem taking $u + o(u)$ bits of space and supporting constant-time rank. Constant-time select within $o(u)$ bits of additional space was added by Munro~\cite{munro1996tables} and Clark~\cite{clark1998compact}. These results were further improved firstly by Pagh~\cite{pagh2001low} (who considered \emph{rank}) and then by Raman et al.~\cite{raman2002succinct} (\emph{rank} and \emph{select}) with structures having the same time complexities and requiring only $\bin + \bigO(u\log\log u/\log u)$ bits of space, where $\bin = \lceil \log \binom{u}{n}\rceil$ is the minimum number of bits required in order to distinguish \emph{any} two size-$n$ subsets of $U$. Finally, P\v{a}tra\c{s}cu~\cite{patrascu2008succincter} reduced the sublinear term to $\bigO(u/polylog(u))$ while retaining constant query times. 
Despite these last results being optimal for big values of $n$, the $o(u)$ term can however be much bigger than $\bin$ (even exponentially) if $n$ is very small. Moreover, even the $\bin$ term is not optimal for \emph{all} instances, and can be improved in many cases of practical interest. To see why this fact holds true, it is sufficient to notice that zero-order entropy compressors encode to the same bit-size \emph{all} size-$n$ subsets $S$ of $U$, without taking advantage of the structure of $S$ (for example, long or regular distances between its elements). This problem motivates the search for more \emph{data-aware} measures able to break the  $\bin$ limit in some cases. One of the most widely known such data-aware measures is \emph{gap}~\cite{bell1993data}, which is defined to be the sum of all bit-lengths of the distances between consecutive elements in $S$. If these distances are not evenly distributed, $gap$ can be much smaller than $\bin$, reaching $10\%$-$40\%$ of $\bin$ in some instances of practical interest~\cite{gupta2007compressed}. 
By using \emph{logarithmic codes} such as Elias $\delta$-encoding~\cite{elias1975universal}, the stream of gaps can be compressed to $gap + o(gap)$ bits, where the $o(gap)$ overhead comes from the prefix property of such codes, needed to unambiguously reconstruct codeword boundaries. 
In~\cite{gupta2007compressed}, Gupta et al. show how to build a FID based on $\delta$-encoding requiring only $gap + \bigO(n\log(u/n)/\log n) + \bigO(n\log\log(u/n))$ bits of space and supporting rank and select in $AT(u,n)\in o((\log\log u)^2)$---this is nearly optimal within that space, see~\cite{andersson2000tight,beame1999optimal}---and $\bigO(\log\log n)$ time, respectively. Other recent works~\cite{makinen2007rank,sadakane2006squeezing} showed that constant-time queries can be supported using $gap + \bigO(n\log\log(u/n)) + o(u)$ bits of space, where the $o(u)$ term is $\bigO(u\log\log u/\sqrt{\log u})$ in~\cite{makinen2007rank} and $\bigO(u\log\log u/\log u)$ in~\cite{sadakane2006squeezing}. 

$gap$ reaches its maximum when all gap lengths are equal. However, it is clear that in this scenario other techniques (e.g. zero-order entropy compression) could be flanked to gap encoding in order to turn this worst-case into a $\bigO(n)$-bits best-case. In this paper we explore the possibility of compressing the stream of gaps $G$ to its zero-order empirical entropy $\Hg$, aiming at obtaining $n\Hg$ as leading term in the space complexity of our structures. 
Similar techniques are already employed in BWT-based text compression algorithms~\cite{burrows1994block}, where runs of zeros in the move-to-front encoding of the BWT are compressed using run-length-encoding followed either by zero-order entropy compression or by logarithmic encoding~\cite{elias1975universal} (runs being mostly dominated by small numbers).
We firstly observe that $n\Hg\in o(gap)$ if gaps are highly compressible, and prove that $n\Hg$ does not exceed $n \log(u/n)+n$ bits. This bound is provably smaller than the zero-order empirical entropy of the set $S$ and of any of its decodable gap-encoded representations.

These considerations suggest that the data-aware measure $n\Hg$ should be preferred to $gap$ in cases where the overhead introduced by the zero-order compressor (e.g. a codebook) is negligible. Our work goes in this direction. First of all, we show a new zero-order compressed representation of bitvectors taking $nH_0(G)+n+\bigO(d\log u) \leq uH_0+n+\bigO(\sqrt{u}\log u)$ bis of space, where $u$ is the length of the bitvector, $H_0$ its zero-order empirical entropy, $n$ the number of bits set, and $d$ the number of \emph{distinct} distances between bits set. $d$ is trivially upper-bounded by $n$ and $\bigO(\sqrt{u})$, and is negligible in many practical cases (e.g. when $S$ is dense or the gaps are evenly distributed).

We finally propose a fully-indexable dictionary that answers \emph{rank} and \emph{select} queries in  $\bigO(\log(u/n)+\log\log u)$ time and whose space occupancy is of $(1+o(1))n\Hg + (3+o(1))n + \bigO((d+\log\log u)\log u)$ bits. In all cases where $\Hg\in\omega(1)$ and $d\geq \log\log u$, this is asymptotically the same space as our new bitvector representation. Moreover, if $S$ is dense enough---$n\in\Omega(u/polylog(u))$---all queries are supported in $\bigO(\log\log n)$ time, which is optimal within this space.


\section{Gap-Encoded Dictionaries}

In this section we will assume that $u-1\in S$, so that each gap corresponds to an element in $S$ (i.e. the element following the gap). If $u-1\notin S$, then we can simply use an extra bit to denote this case and encode the final gap length separately. We will moreover assume that $n\leq u/2$. Logarithms are taken in base 2, unless differently specified. In \emph{gap encoding}, we represent the set $S=\{s_1,...,s_n\}\subseteq \{0,...,u-1\} = U$, $s_1<s_2<...<s_n$ as the stream of gaps $g_1,...,g_n$, where $g_1 = s_1+1$ and $g_i = s_i-s_{i-1}$ for $i>1$. In order to reduce space occupancy of the stream, variable-length encoding can be used to encode each of the $g_i$. The data-aware measure $gap(S)$ is defined as $gap(S) = \sum_{i=1}^{n}\big( \lfloor \log g_i \rfloor + 1\big)$, that is, the total number of bits required in order to store all $g_i$'s using the \emph{minimum} number of bits to represent each gap. When clear from the context, we will simply write $gap$ instead of $gap(S)$.
Clearly, $S$ cannot be represented using only $gap$ bits since we need additional information in order to make the stream uniquely decodable. We adopt a notation similar to~\cite{gupta2007compressed} and indicate with $Z_\mathcal C(S)$---or simply $Z_\mathcal C$ when clear from the context---the decoding overhead (in bits) introduced by the coding scheme $\mathcal C$. If we use a separate bitvector $B$ marking with a 1 the beginning of each code, then we obtain $Z_B=gap$. Another solution is to use \emph{logarithmic codes} such as Elias $\gamma$ or $\delta$-encoding~\cite{elias1975universal}. In $\gamma$-encoding, we encode $\lfloor \log g_i \rfloor +1$ in unary, followed by the $\lfloor \log g_i \rfloor$-bits binary representation of $g_i$ without the most significant $1$. Then, $Z_\gamma = gap-n$. A better solution is $\delta$-encoding, where we encode with $\gamma$ the number $\lfloor \log g_i \rfloor +1$, followed by the $\lfloor \log g_i\rfloor$-bits binary representation of $g_i$ without the most significant 1. Then, $Z_\delta = 2\sum_{i=1}^n \lfloor \log(\lfloor \log g_i\rfloor+1) \rfloor$ bits. 
$\log$ being a concave function, the worst-case of $gap$ occurs when $g_1=g_2=...=g_n = u/n$ (by Jensen's inequality), yielding the upper bounds $gap\leq n\log(u/n)+n$ and $Z_\delta\leq 2n\log(\log(u/n)+1)$. Then, one can prove the following (for the original proof, see~\cite{grossi2004indexing}):
\begin{lemma}\label{lemma:gap less bin}
$gap \leq \bin$ if $n\leq u/2$.
\end{lemma}
\begin{proof}
The claim follows directly from $gap\leq n\log(u/n)+n$ and from the fact that $\bin = n\log(u/n) + n\log e - \Theta(n/u) + \bigO(\log n)$ if $n\leq u/2$\qued
\end{proof}

Moreover, let $H_0(S) = \frac{n}{u}\log(\frac{u}{n}) + \frac{u-n}{u}\log\frac{u}{u-n}$ be the zero-order empirical entropy of the set $S$. Since $\bin\leq uH_0(S)$, we have that:
\begin{corollary}\label{corollary gap < uH0}
$gap \leq uH_0(S)$ if $n\leq u/2$.
\end{corollary}

The above inequalities are important as they show that gap encoding can never perform worse than zero-order entropy compression. On the other hand, experiments show~\cite{gupta2007compressed} that $gap$ can be significantly smaller than $\bin$ for many cases of interest, thus motivating its use in practical applications. In the following section we take one step forward, exploring what happens when we treat $S$ as a sequence on the alphabet $\{g_1,...,g_n\}$ and then apply zero-order entropy compression to it.

\subsection{A Compressed-Gap Data-Aware Measure}

$gap$ reaches its worst-case of $n\log(u/n)+n$ bits when all gaps have the same length. However, it is clear that entropy compression should turn this worst-case scenario into a best-case, since the zero-order empirical entropy of such a configuration is equal to 0.
More formally, let's consider the following representation $G$ of $S$. We define $G$ to be the sequence $g_1g_2...g_n\in\Sigmagap^n$, where $\Sigmagap = \{g_1,g_2, ..., g_n\}$. Let moreover $d = |\Sigmagap|$ be the alphabet size and $f(s)=occ(s)/n$, $s\in\Sigmagap$, be the empirical relative frequency of $s$ in $G$, where $occ(s)$ is the number of occurrences of $s$ in $G$. We define the \emph{zero-order empirical entropy of the gaps} $\Hg$ to be
\begin{definition}\label{def:gapentropy}
$\Hg = - \sum_{s\in\Sigmagap} f(s)\log\left(f(s)\right)$
\end{definition}

$n\Hg$ is the minimum number of bits output by any compressor that encodes $G$ assigning a unique code to each symbol in $\Sigmagap$. First of all, we observe that $n\Hg$ can be significantly smaller than $gap$: if $g_1=g_2=...=g_n = u/n$, then $n\log(u/n) \leq gap \leq n\log(u/n) + n$ and $n\Hg =0$. Moreover, $n\Hg$ is never worse than the length of any decodable gap-compressed sequence:
\begin{lemma}\label{th:compressed gap vs gap}
$n\Hg \leq gap + Z_\mathcal C$, where $\mathcal C$ is any prefix coding scheme.
\end{lemma}
\begin{proof}
Follows directly from the fact that no prefix code can compress $G$ in less than $n\Hg$ bits.\qued
\end{proof}

Using Lemma \ref{th:compressed gap vs gap} and the bounds for $gap$ and $Z_\delta$ derived in the previous section, one can obtain $\Hg \leq \log(u/n) + 2\log(\log(u/n)+1)+ 1$. With the following theorem we show a much stronger upper bound:

\begin{theorem}\label{th:worst case Hg}
$\Hg \leq \log(u/n) +1$
\end{theorem}
\begin{proof}
We want to compute 
$$
\max_{\Sigmagap \subseteq \mathbb{N}_{> 0}}\max_{f:\Sigmagap \rightarrow \mathbb R^+} \Hg
$$
where the alphabet $\Sigmagap$ and the empirical frequency function $f$ must satisfy: 
\begin{equation}\label{eq:condition}
n\sum_{s\in\Sigmagap}f(s)\cdot s = u
\end{equation}

Let $d=|\Sigmagap|$. From Definition \ref{def:gapentropy} and from the concavity of $\log$, we have that $\Hg$ reaches its maximum $\Hg=\log d$ when all frequencies are equal, i.e. $f(s)=d^{-1}$ for all $s\in\Sigmagap$. We thus have
$$
\max_{\Sigmagap \subseteq \mathbb{N}_{> 0}}\max_{f:\Sigmagap \rightarrow \mathbb R^+} \Hg = \max_{\Sigmagap \subseteq \mathbb{N}_{> 0}, f(s)=d^{-1},\ s\in\Sigmagap} \log d
$$

In order to maximize $\log d$, we now have to find $\Sigmagap$ of maximum cardinality that satisfies condition (\ref{eq:condition}). It is easy to see that $\Sigmagap=\{1,..., d\}$ minimizes $\sum_{s\in\Sigmagap} s = \sum_{i=1}^{ d} i =  d( d+1)/2$. Since, moreover, $f(s)= d^{-1}$ for all $s\in\Sigmagap$, we can rewrite (\ref{eq:condition}) as $n d^{-1}\big( d( d+1)/2 + k\big) = u$, where $k\geq 0$. Solving in $ d$, we obtain the set of solutions
$$
\mathcal Z=\left\{ \big(b\pm \sqrt{b^2-8kn^2}\big)/(2n)\ |\ b=2u-n\ \wedge\ k\geq 0 \right\}
$$
for which we have $\max \mathcal Z = (2u-n)/n$ when $k=0$. This implies that $\Sigmagap=\{1,...,(2u-n)/n\}$ and $f(s)=n/(2u-n)$ for all $s\in\Sigmagap$ maximize $\Hg$. Our claim follows:
$$
\Hg \leq \log  d \leq \log(2u/n) = \log(u/n) + 1
$$
\qued
\end{proof}

Interestingly, the two measures $gap$ and $n\Hg$ are upper-bounded by the same quantity $n\log(u/n) + n$. This is not a trivial result since, differently from $n\Hg$, $gap$ does not include information needed to reconstruct unambiguously codeword boundaries (even though $n\Hg$, in turn, does not include information---e.g. a codebook---needed to decode codewords). Using the same arguments of Lemma \ref{lemma:gap less bin} and Corollary \ref{corollary gap < uH0}, we can moreover derive the bounds:
\begin{corollary}\label{corollary1}
$n\Hg \leq \bin \leq uH_0(S)$ if $n\leq u/2$
\end{corollary}

The pair $\langle U,S\rangle$ can be represented as a length-$u$ bitvector $B$ with $n$ bits set. Let $H_0=H_0(S)$ be the zero-order entropy of $B$ and $d$ be the number of \emph{distinct} distances between bits set in $B$. Then:
\begin{corollary}\label{zero-order bitvector}
There exists a zero-order compressed representation of $B$ taking $n(H_0(G)+1)+\bigO(d\log u) \leq uH_0+n+\bigO(d\log u)$ bits of space.
\end{corollary}
\begin{proof}
Can be easily obtained by compressing the gap sequence with Huffman-encoding and by applying Corollary \ref{corollary1}.
\end{proof}

Note that the number $d$ of \emph{distinct} distances between bits set of $B$ is trivially upper-bounded by $n$ and $\bigO(\sqrt{u})$\ \footnote{Assume, by contradiction, that $ d \in \omega(\sqrt{u})$. Then, the set $\Sigmagap$ of gaps that minimizes $\sum_{s\in\Sigmagap}s$ is $\Sigmagap=\{1,..., d\}$, for which we obtain $\sum_{s\in\Sigmagap}s = \Theta( d^2) = \omega(u)$. This is an absurd since the sum of all gaps cannot exceed $u$.}.

%
%

\section{A Compressed-Gap FID}

Let us now turn our attention to fully-indexable dictionary data structures. Our aim is to obtain a structure that takes asymptotically the same space as the representation described in Corollary \ref{zero-order bitvector}. 

Our strategy is the following: we use Elias $\delta$-encoding and exploit its property of being an \emph{asymptotically optimal universal code}~\cite{elias1975universal} to encode the gap stream in $(1+o(1))n\Hg+n$ bits. We then build a two-levels structure atop of this representation to support rank and select queries. We adopt an approach similar to~\cite{gupta2007compressed} and firstly describe a binary-searchable dictionary (BSD) that supports all queries in $\bigO(\log u)$ time. The BSD is finally used as building block for our final structure, which improves all query times to $\bigO(\log(u/n)+\log\log u)$ within the same space.

Let $\Sigmagap$ and $f:\Sigmagap\rightarrow \mathbb R^+$ be the set of all gap lengths and the empirical frequencies associated with the gap stream, respectively, and consider an (arbitrary) ordering of the symbols $ord:\Sigmagap \rightarrow \{1,..., d\}$, $ d=|\Sigmagap|$ (i.e. a bijection) such that if $ord( g_i ) < ord( g_j )$ then $f(g_i) \leq f(g_j)$ for all $g_i, g_j \in \Sigmagap$. Let $\delta(x),\ x>0$ be the Elias $\delta$ code of the integer $x$. Then, we associate the code $code(g_i) = \delta(ord(g_i))$ to each gap length $g_i\in\Sigmagap$. Being $\delta$ an asymptotically optimal universal code~\cite{elias1975universal}, the bit length $l$ of the compressed stream $code(g_1)...code(g_n)$ is at most $(1+o(1))n\Hg + n$ bits\footnote{Even when $\Hg=0$, with $\delta$-encoding we spend \emph{at least} 1 bit per symbol, thus the additional $n$ term. The $o(n\Hg)$ term comes from overhead introduced by $\delta$-encoding, and in the \emph{worst} case ($n$ distinct gaps, $\Hg \in\Theta(n\log n)$) equals $\Theta(n\log\log n)$ bits.}.
In the following we assume to work under the word RAM model with word size $\Theta(\log u)$ bits, so that we can extract any $\bigO(\log u)$-bits block from a plain bitvector in constant time.
We store the bit representations of the compressed gaps sequentially in a bitvector $C[0,...,l-1] = code(g_1)...code(g_n)$. An additional array $D[1,..., d]$ defined as $D[i] = ord^{-1}(i)$ (the codebook) is moreover built to permit the decoding of codewords. 
Note that, given the starting position of $code(g_i)$, $0\leq i < n$, in the bitvector $C$, we can extract and decode $code(g_i)=\delta(ord(g_i))$ in $\bigO(1)$ time: firstly, we need to decode the $\gamma$-prefix of $\delta(ord(g_i))$. This can be done in $\bigO(1)$ time using two universal tables of $\bigO(2^{\log\log u}\log\log u) = \bigO(\log u\log\log u)$ bits each (one for the unary prefix and the other for the rest of the $\gamma$-prefix of the code). This gives us (i) the bit-length of the $\gamma$-prefix of $\delta(ord(g_i))$, and (ii) the bit-length of $ord(g_i)$ (without the most significant bit). We can then extract the bits of $ord(g_i)$ and access $D[ord(g_i)] = g_i$ in constant time. To improve readability, in the next sections we will implicitly make use of this strategy and---provided that we know the starting position of $code(g_j)$ in $C$---say \emph{read gap $g_j$} instead of \emph{extract and decode $code(g_j)$}.


\subsection{A Binary-Searchable Dictionary}\label{sec: BSD}

We divide the elements of $S=\{s_1,...,s_n\}$ into blocks of size $t=\lceil \log u\rceil$ (we assume for clarity of exposition that $t$ divides $n$; the following arguments can be easily adapted to the general case). For each block $\{s_{it+1},...,s_{(i+1)t}\}$, $i=0,...,n/t-1$, we store explicitly the smallest element $s_{it+1}$ and a pointer to the beginning of $code(g_{it+2})$ in the bitvector $C$ \footnote{We point to $code(g_{it+2})$ instead of $code(g_{it+1})$ because $s_{it+1}$ is explicitly stored. As a matter of fact, we can avoid storing $code(g_{it+1})$ in $C$.}. These structures are sufficient to obtain our BSD. $select(S,i)$, $0\leq i<n$, is implemented by accessing the $\lfloor i/t \rfloor$-th block and reading $i\mod t < t$ gaps in $C$ starting from $g_{\lfloor i/t \rfloor t+2}$. Then, 
$$
select(S,i) = s_{\lfloor i/t \rfloor t+1} + \sum_{j=\lfloor i/t \rfloor t+2}^{i+1} g_j
$$
$rank(S,x)$, $x\in U=\{0,...,u-1\}$, is implemented by binary-searching the blocks according to explicitly stored elements $s_{it+1}$, $i=0,...,n/t-1$, and then by extracting gaps in the block of interest until we reach element $x$. More formally, let $0 \leq i \leq n/t-1$ be the biggest integer (if any) such that $s_{it+1}\leq x$. $i$ can be found by binary search in $\bigO(\log u)$ time. If such an integer does not exist, then $rank(S,x)=0$. Otherwise, let $1\leq j < t$ be the smallest integer such that $q = s_{it+1} + \sum_{h=1}^j g_{it+1+h} \geq x$. $j$ can be found by linear search in $\bigO(t) = \bigO(\log u)$ time. Then, 
$$
rank(S,x) = \left\{ \begin{array}{ll} it + j+1 & if\ q=x\\ it+j & if\ q>x\end{array} \right.
$$
The bit-length of $C$ is at most $\bigO(n\log u)$, so a pointer to $C$ takes $\log n + \log\log u +\bigO(1) \leq \log u+ \log\log u+\bigO(1)$ bits. It follows that for each block we spend $2\log u+\log\log u+\bigO(1)$ bits (one element $s_{it}$ and a pointer to $C$), so the blocks take overall $(2\log u+\log\log u+\bigO(1)) \cdot n/\log u = 2n+o(n)$ bits. 
We obtain:
\begin{lemma}\label{th:BSD}
Let $d$ be the number of distinct gap lengths between elements in $S$. The binary-searchable dictionary described in section \ref{sec: BSD} occupies $(1+o(1))n\Hg + (3+o(1))n + \bigO((d+\log\log u)\log u)$ bits of space and supports rank and select queries in $\bigO(\log u)$ time.
\end{lemma}

Note that the size of the proposed BSD can be \emph{exponentially} smaller than $u$ if $S$ is sparse. In the next section we show how to obtain $\bigO(\log(u/n)+\log\log u)$-time queries without asymptotically increasing space usage.

\subsection{A Fully-Indexable Dictionary}\label{sec: FID}

Let $v=\lceil u\log^2 u/n\rceil$. The idea is to divide $U$ into blocks of $v$ elements, and store a BSD for each block. 

We build a constant-time rank and select succinct bitvector $V[0,...,\lceil u/v\rceil -1]$ defined as $V[i] = 1$ if and only if $S\cap \{iv,...,(i+1)v-1\} \neq \emptyset$. Additionally, one array $R[0,...,\lceil u/v\rceil -1]$ stores sampled \emph{rank} results: $R[0]=0$ and $R[i] = rank(S,iv-1)$ for $i>0$. We build a binary-searchable dictionary $BSD(i)$ for each set $S_i = \{ x-iv\ |\ x\in S\cap \{iv,...,(i+1)v-1\}\},\ i=0,...,\lceil u/v\rceil -1$, where we use the same codebook $D$ for all the BSD structures (i.e. $D$ is computed according to all gaps $g_1, ..., g_n$). Note that there may exist a set $S_i$ (or more than one) such that its first gap does not belong to $\{g_1,...,g_n\}$. This happens each time an element $s_i$ is the first of its block $b = \lfloor s_i/v\rfloor>0$, the gap $g_i$ overlaps blocks $b$ and $b-1$, and $s_i-b\cdot v+1\notin \{g_1,...,g_n\}$. However, by construction of the BSD data structure (see previous section), the first gap in $S_i$ is never used (since we store the smallest element of $S_i$ explicitly), so this event does not affect overall gap frequencies nor space requirements of the array $D$. Finally, one array $SEL[0,...,\lceil n/t \rceil-1]$, where $t=\lceil \log^2 u\rceil$, stores the (number of the) block containing $s_{it+1}$:  $SEL[i] = \lfloor s_{it+1}/v\rfloor$, for $i=0,...,\lceil n/t \rceil-1$.

Using the above described structures, we can now show how to efficiently solve queries. $rank(S,x)$, $x\in U=\{0,...,u-1\}$, is implemented by accessing the $\lfloor x/v \rfloor$-th block and calling $rank$ on $BSD(\lfloor x/v \rfloor)$. More formally, 
$$
rank(S,x) = R[\lfloor x/v \rfloor] + rank(S_{\lfloor x/v \rfloor},x\mod v)
$$
where $rank(S_{\lfloor x/v \rfloor},x\mod v)$ is called on the structure $BSD(\lfloor x/v \rfloor)$. Rank is thus solved in $\bigO(\log v) = \bigO(\log(u/n)+\log\log u)$ time. To solve $select(S,i)$, we firstly find by binary search the block containing $s_{i+1}$, and then call $select$ on the corresponding BSD. More in detail, let $q_l = SEL[ \lfloor i/t\rfloor ]$ and $q_r = SEL[ \lfloor i/t\rfloor +1]$ if $\lfloor i/t\rfloor +1<\lceil n/t\rceil$, $q_r=q_l$ otherwise. By construction of $SEL$, the block containing element $s_{i+1}$ is one of $q_l,q_l+1,...,q_r$. Note that the number $q_r-q_l+1$ of blocks of interest can be arbitrary large since there may be an arbitrary number of \emph{empty} blocks among them. However, at most $t$ of them will contain \emph{at least} one element (by construction of $SEL$). Then, we can perform binary search only on the blocks marked with a 1 in the array $V$: during binary search we access blocks at positions of the form $select(V,j)$ (note: this is a constant-time select performed on the bitvector $V$), starting with the range $j\in [rank(V,q_l)-1, rank(V,q_r)-1]$. Binary search is performed according to partial ranks (array $R$). Let $q_l\leq q_m \leq q_r$ be the biggest integer such that $R[q_m] \leq i < R[q_m+1]$ (if $q_m+1\geq \lceil u/v \rceil$ then simply ignore the upper bound in the previous inequality). According to the above considerations, $q_m$ can be found in $\bigO(\log t) = \bigO(\log\log u)$ time using binary search. We can solve $select(S,i)$ as follows:
$$
select(S,i) = q_m\cdot v + select(S_{q_m}, i - R[q_m] )
$$
where $select(S_{q_m}, i - R[q_m] )$ is called on the structure $BSD(q_m)$. $select$ is thus solved on our FID in $\bigO(\log v)+\bigO(\log\log u) = \bigO(\log(u/n)+\log\log u)$ time. 

Bitvector $V$ takes $(1+o(1))u/v = (1+o(1))n/\log^2 u=o(n)$ bits. Arrays $R$ and $SEL$ take $\log u \cdot u/v = n/\log u = o(n)$ and $\log u \cdot n/t = n/\log u = o(n)$ bits of space, respectively. Finally, all BSD data structures take overall $(1+o(1))n\Hg + (3+o(1))n$ bits, and the codebook $D$ and the universal tables take $\bigO((d+\log\log u)\log u)$ bits. We can state our final result:
\begin{theorem}\label{theorem: FID}
Let $d$ be the number of distinct gap lengths between elements in $S$. The FID described in section \ref{sec: FID} takes $(1+o(1))n\Hg + (3+o(1))n + \bigO((d+\log\log u)\log u)$ bits of space and supports rank and select queries in $\bigO(\log(u/n)+\log\log u)$ time.
\end{theorem}

The result stated in Theorem \ref{theorem: FID} improves the space of~\cite{makinen2007rank,sadakane2006squeezing}, reducing both leading and $o(u)$ terms from $gap+\bigO(n\log\log(u/n))$ and $u\log\log u/\log u$ bits to $(1+o(1))n\Hg + (3+o(1))n$ and $\bigO((d+\log\log u)\log u)\subseteq \bigO(\sqrt{u}\log u)$ bits, respectively. 
This improvement comes at the price of a $\bigO(\log(u/n)+\log\log u)$ slowdown in all query times.  Notice that we cannot apply the general technique proposed by M\"akinen and Navarro in~\cite{makinen2007rank} in order to obtain $\bigO(1)$ query times since $code()$ does not (always) satisfy $|code(x)| \in \bigO(\log x)$ (this is one of the properties characterizing \emph{random access self-delimiting codes}~\cite{makinen2007rank}). An interesting line of research would be to envision a broader class of codes (including $code()$) for which we can describe a general technique guaranteeing constant-time queries.

\section{$\Hg$ in practice}

In order to assess also in practice the differences between the above discussed measures, we adopted the approach of~\cite{grossi2004indexing} and simulated several sets, computing for each of them the number of bits per item required by $gap$, $gap + Z_\delta$, $uH_0(S)$, $n\Hg$, $n\Hg+Z_\delta$, and $n\Hg+Z_\delta+CB$, where the last two measures refer to $\Hg$ plus the overhead introduced by $\delta$-encoding (i.e. encoding $g_1,...,g_n$ as described in the previous section) and by the codebook size (CB).

Gaps were generated according to uniform (Table \ref{tab:uniform simulation}) and binomial (Table \ref{tab:binomial simulation}) distributions. Table \ref{tab:uniform simulation} reports the same experiment performed in~\cite{grossi2004indexing} (except from the facts that we use $\delta$ instead of $\gamma$ and we do not consider RLE), updated with our measure $n\Hg$. 
As expected, in this case $n\Hg$ performs slightly worse than $gap$ when taking into account all encoding overheads (columns 3 and 7). This can be explained by the fact that gaps are uniform, thus making $gap + Z_\delta$ and  $n\Hg+Z_\delta$ (without the codebook) almost equivalent. An interesting fact---in accordance with Theorem \ref{th:worst case Hg}---is that, even this being its worst case, $n\Hg$ is always smaller (by about 0.5 bits per item) than $uH_0(S)$.

The advantages of using $n\Hg$ become evident when non-uniform distributions are used. 
Table \ref{tab:binomial simulation} reports the results on binomially-distributed gaps\footnote{We chose a binomial distribution in order to model a scenario in which gap lengths are accumulated around a value $\mu\gg 0$ (in this case, $\mu$ is the mean). Intuitively, in this case $gap$ does not perform well because small numbers are not frequent.}. As expected, in this case our measure considerably improves on $gap$: if the two techniques are compared while taking into account all encoding overheads (columns 3 and 7), our strategy requires about $58\%$ the space of $gap$ encoding.

\begin{table}
\centering
\begin{tabular}{|c|c|c|c|c|c|c|}\hline
$\log(max\_gap)$ & $gap$ & $gap + Z_\delta$ & $uH_0(S)$ & $n\Hg$ & $n\Hg+Z_\delta$ & $n\Hg+Z_\delta+CB$ \\\hline
1	&1.66717	&3.00151	&2.00103	&1.58496	&2.99842	&2.99848\\
2	&2.20164	&3.80142	&2.75854	&2.32191	&3.79349	&3.79364\\
3	&2.77733	&5.00151	&3.61667	&3.16987	&4.98418	&4.98454\\
4	&3.47452	&6.53906	&4.5389	&4.08735	&6.50696	&6.50781\\
5	&4.2771	&7.79638	&5.50097	&5.04417	&7.75575	&7.75773\\
6	&5.15079	&8.90439	&6.48606	&6.02187	&8.8685	&8.87305\\
7	&6.09095	&10.0028	&7.4809	&7.01044	&9.94679	&9.95711\\
8	&7.04186	&11.9893	&8.48908	&8.00377	&11.889	&11.9122\\
9	&8.02066	&13.4915	&9.50168	&8.99923	&13.3703	&13.4216\\
10	&9.01571	&14.7531	&10.5266	&9.99358	&14.5752	&14.6879\\
11	&10.0076	&15.8755	&11.5554	&10.9857	&15.661	&15.9068\\
12	&11.0103	&16.9465	&12.599	&11.9707	&16.6565	&17.1892\\
13	&12.0031	&17.9701	&13.6584	&12.94	&17.5894	&18.7364\\
14	&13.0009	&18.9844	&14.7359	&13.8789	&18.4625	&20.9157\\
15	&13.996	&19.9873	&15.839	&14.7427	&19.2538	&24.2575\\\hline
\end{tabular}

\

\caption{Comparison between $gap$, $gap + Z_\delta$, $uH_0(S)$, $n\Hg$, $n\Hg+Z_\delta$ (i.e. accounting for the $\delta$ overhead per symbol), and $n\Hg+Z_\delta+CB$ (i.e. accounting for the $\delta$ and codebook CB overhead per symbol) on randomly-generated sets. Gaps between the $n$ items ($n$ affects only the variance of the results; we used $n=10^5$) are uniformly distributed in the interval $[1,max\_gap]$. All columns except the first report the number of bits per item required by each method.}\label{tab:uniform simulation}

\begin{tabular}{|c|c|c|c|c|c|c|}\hline
$\log(max\_gap)$ & $gap$ & $gap + Z_\delta$ & $uH_0(S)$ & $n\Hg$ & $n\Hg+Z_\delta$ & $n\Hg+Z_\delta+CB$ \\\hline
1	&1.74989	&3.24967	&2.22939	&1.50052	&2.50156	&2.50162\\
2	&2.25085	&4.12525	&3.16331	&2.03377	&3.00555	&3.0057\\
3	&2.88491	&4.94587	&4.18044	&2.5445	&3.49472	&3.49508\\
4	&3.77183	&7.31887	&5.22493	&3.04741	&4.094	&4.09485\\
5	&4.70015	&8.69979	&6.27376	&3.54494	&4.82176	&4.82326\\
6	&5.64788	&9.64788	&7.31532	&4.04711	&5.61441	&5.61679\\
7	&6.60309	&10.6031	&8.3491	&4.54742	&6.3782	&6.3822\\
8	&7.57464	&12.7239	&9.37466	&5.04947	&7.08812	&7.09424\\
9	&8.55226	&14.5523	&10.3937	&5.54834	&7.75208	&7.76178\\
10	&9.53716	&15.5372	&11.4078	&6.04518	&8.33989	&8.35386\\
11	&10.5229	&16.5229	&12.4178	&6.54489	&8.93035	&8.95219\\
12	&11.516	&17.516	&13.425	&7.04343	&9.56187	&9.59411\\
13	&12.5135	&18.5134	&14.4301	&7.54485	&10.3296	&10.3775\\
14	&13.5082	&19.5082	&15.4338	&8.03851	&11.1441	&11.2149\\
15	&14.5084	&20.5084	&16.4364	&8.53758	&11.9996	&12.1044\\\hline
\end{tabular}

\

\caption{Comparison between $gap$, $gap + Z_\delta$, $uH_0(S)$, $n\Hg$, $n\Hg+Z_\delta$, and $n\Hg+Z_\delta+CB$ on randomly-generated sets. Gaps between the $n$ items ($n=10^5$) are binomially distributed in the (shifted) interval $[1,max\_gap]$ with success probability $p=1/2$. All columns except the first report the number of bits per item required by each method.}\label{tab:binomial simulation}
\end{table}

\section{Conclusions}

In this paper we introduced $\Hg$, a new data-aware measure based on the idea of compressing the gaps between elements of a set $S\subseteq \{0,...,u-1\}$. We provided new theoretical upper-bounds for this measure, and showed that in practice---if the gap stream is compressible---$\Hg$ considerably improves space usage of gap encoding techniques combined with logarithmic codes such as Elias $\delta$-encoding. Finally, we proposed a new zero-order representation of bitvectors based on our new measure and a compressed-gap fully-indexable dictionary supporting fast queries and taking small space in addition to $n\Hg$.

As expected, simulations confirmed that the proposed compressed-gap measure is particularly convenient in situations where the gaps follow a non-uniform distribution or they are dominated mainly by large numbers. The main drawback of $n\Hg$ seems to be the overhead introduced by the zero-order compressor, which in our solution is of $\Theta(\sqrt u\log u)$ bits in the worst case. However, in some practical applications this overhead---being proportional to the number $d$ of \emph{distinct} gap lengths---is expected to be negligible with respect to the overall structure size. One example of such an application is run-length compression of the BWT of highly repetitive text collections (e.g. genome variants), where run lengths are expected to scale linearly with the \emph{number of documents} in the collection~\cite{makinen2010storage,siren2012compressed}. 


We plan to implement our FID and test it against state-of-the-art practical gap-encoded bitvector representations (e.g. \texttt{sd\_vector} of SDSL\cite{gbmp2014sea}). Notice that in practice Huffman-compression of the gaps should be preferred to universal delta-encoding, as the additional overhead is much smaller (i.e. we can remove the $o(n\Hg)$ term). Our FID could find a first application in repetition-aware self-indexing, e.g. by using it as building block of a more space-efficient run-length compressed suffix array (RLCSA\cite{siren2012compressed}).

\bibliographystyle{splncs03}
\bibliography{cGAP}

\end{document}